\documentclass[twocolumn,pra,aps,showpacs,superscriptaddress,floatfix,showkeys,nofootinbib,10pt]{revtex4-2}

\usepackage[T1]{fontenc}
\usepackage{amsmath,amsfonts,amssymb,amsthm}
\usepackage{graphicx}
\usepackage{subfigure}

\usepackage{tikz}
\usetikzlibrary{arrows.meta}

\newcommand{\ket}[1]{ | \, #1 \rangle} \newcommand{\bra}[1]{ \langle #1 \, |} 
\newcommand{\proj}[1]{\ket{#1}\bra{#1}} 
\newcommand{\kb}[2]{\ket{#1}\bra{#2}}

\newcommand{\Ab}[1]{ \left| #1 \, \right|} 
\newcommand{\be}{\begin{equation}} \newcommand{\ee}{\end{equation}}
\newcommand{\ba}{\begin{aligned}} \newcommand{\ea}{\end{aligned}}

\DeclareMathOperator{\Tr}{Tr}

\DeclareRobustCommand\openone{\leavevmode\hbox{\small1\normalsize\kern-.33em1}}%

\newtheorem{dfn}{Definition}

\newtheorem{cor}{Corollary}
\newtheorem{lemma}{Lemma}

\begin{document}

\title{Quantum Strategies for Rendezvous and Domination Tasks on Graphs with Mobile Agents}

\author{Giuseppe Viola} \email{giuseppe.viola.res@gmail.com}
\affiliation{International Centre for Theory of Quantum Technologies, University of Gda\'{n}sk, Wita Stwosza 63, 80-308 Gda\'{n}sk, Poland}

\author{Piotr Mironowicz} \email{piotr.mironowicz@gmail.com}
\affiliation{International Centre for Theory of Quantum Technologies, University of Gda\'{n}sk, Wita Stwosza 63, 80-308 Gda\'{n}sk, Poland}
\affiliation{Department of Physics, Stockholm University, S-10691 Stockholm, Sweden} 
\affiliation{Department of Algorithms and System Modeling, Faculty of Electronics, Telecommunications and Informatics, Gda\'{n}sk University of Technology, Poland}

\date{\today}

\begin{abstract}
	This paper explores the application of quantum non-locality, a renowned and unique phenomenon acknowledged as a valuable resource. Focusing on a novel application, we demonstrate its quantum advantage for mobile agents engaged in specific distributed tasks without communication. The research addresses the significant challenge of rendezvous on graphs and introduces a new distributed task for mobile agents grounded in the graph domination problem. Through an investigation across various graph scenarios, we showcase the quantum advantage. Additionally, we scrutinize deterministic strategies, highlighting their comparatively lower efficiency compared to quantum strategies. The paper concludes with a numerical analysis, providing further insights into our findings.
\end{abstract}


\maketitle

\section{Introduction}

Quantum information has emerged as a valuable asset in a multitude of tasks inherent to information and communication technologies. Remarkably, quantum cryptography stands out as an example, where the harnessing of quantum entanglement has paved the way for innovative protocols ensuring secure key exchange and robust randomness certification~\cite{pirandola2020advances}. Despite these achievements, the phenomena of Bell inequalities and non-locality, which defy the boundaries of correlations allowed in classical physics~\cite{horodecki2009quantum,brunner2014bell}, have remained largely unexplored for other pragmatic endeavors. In particular, their potential for distributed
tasks in facilitating the coordination of actions among distributed agents with a shared goal, even in situations with limited or no communication, remains almost untapped~\cite{cao2012overview}, with preliminary trials in~\cite{muhammad2014quantum,mironowicz2023entangled}.

In the realm of quantum information, the accomplishments have been particularly pronounced in quantum cryptography~\cite{pirandola2020advances}, where the manipulation of quantum entanglement has enabled the creation of novel protocols guaranteeing the confidentiality and integrity of exchanged keys, as well as certifying the security of random data~\cite{pironio2010random}. The potential of quantum phenomena achieved with the aid of entanglement has been illustrated by multiple examples of so-called non-local games~\cite{cirel1980quantum,cleve2004consequences,russo2017extended}. A prominent example of them are the XOR-games~\cite{regev2015quantum}, and generalizations thereof~\cite{ramanathan2016generalized,russoAndWatrous2017extended,luo2018nonlocality,luo2019nonlocal}. However, amidst these striking achievements, the captivating domains of Bell inequalities and non-locality~\cite{bell1964einstein}, which intrinsically challenge the limitations imposed by classical physics on correlations, have not been widely applied to other practical tasks. A notable area that remains unexplored is their potential application in orchestrating the collaboration of distributed agents striving towards a collective objective, even in scenarios where communication is restricted or absent, like the rendezvous task~\cite{alpern1995rendezvous}.

This paper delves into the realm of quantum entanglement and its application in coordinating distributed tasks on graphs. The groundwork is laid with two fundamental problems of graph theory: rendezvous and graph domination~\cite{allan1978domination}. These problems serve as the foundation for exploring the potential benefits of quantum entanglement in achieving coordinated actions among distributed agents striving to achieve a common goal on graphs. We continue and extend the investigations initiated in~\cite{mironowicz2023entangled} where for the first time it was shown how to exploit quantum resources for the rendezvous task.

To tackle this topic, we use the semi-definite programming (SDP)~\cite{Skrzypczyk2023,mironowicz2023semi,tavakoli2023semidefinite}, in particular with two prominent methods: the Navascués-Pironio-Acín method (NPA)~\cite{navascues2007bounding,navascues2008convergent} and the see-saw method~\cite{pal2010maximal}. These techniques offer valuable insights into the quantum non-locality, which challenges the boundaries set by classical physics. We explore how quantum entanglement enables distributed agents to achieve rendezvous with no communication, surpassing classical constraints. Then, the paper introduces the domination distributed task for agents. We are this way laying out the framework for incorporating quantum entanglement as a resource for enhanced coordination. To demonstrate the practical implications of the findings, the paper presents a series of numerical calculations and examples involving various graph structures. These numerical simulations showcase the effectiveness of quantum entanglement in scenarios of distributed actions of agents, illustrating the potential advantages of using quantum strategies over classical ones.

The organization of this paper is the following. In sec.~\ref{sec:tasks} we discuss the two relevant problems of graph theory, viz. the task of rendezvous of mobile agents, and the problem of domination number in graphs. We show how to turn the latter problem into a new task for mobile agents. Then we discuss multipartite probability distributions in classical and quantum physics and show how to express distributed tasks as so-called Bell, or non-local, games investigated in quantum theory. In sec.~\ref{sec:methods} we briefly describe the numerical methods used in this paper based on SDP. Next, in sec.~\ref{sec:theorems} we provide analytical results about deterministic strategies that can be used by the agents for the tasks. The results showing the quantum advantage of the tasks are presented in sec.~\ref{sec:Results}. We discuss the obtained results in sec.~\ref{sec:discuss} and conclude in sec.~\ref{sec:conclusions}.

\section{Distributed Tasks for Agents}
\label{sec:tasks}

We will now describe the tasks for mobile agents investigated in this work. The rendezvous task is discussed in sec.~\ref{sec:rendezvousDef}. In sec.~\ref{sec:dominationDef} we cover the domination number of graphs and then propose a new distributed task inspired by it. The examples of graphs considered in this work are given in sec.~\ref{sec:graphs} together with definitions of selected terms used in describing mobile agents. In sec.~\ref{sec:probs} we introduce classical and quantum probability distributions and Bell games. Finally, in sec.~\ref{sec:tasks_as_games} we show how to relate distributed tasks with Bell games.

\subsection{Rendezvous on Graphs}
\label{sec:rendezvousDef}

The rendezvous problem, as formulated by Steve Alpern~\cite{alpern1995rendezvous}, is a mathematical problem that deals with the challenge of two or more mobile agents trying to meet at a specific location without any form of communication. In its simplest form, the rendezvous problem involves two mobile agents, referred to as players or parties, that are initially located in different positions within a known environment. The objective is for both players to reach a common meeting point simultaneously. The key constraint is that the players cannot communicate with each other and have no knowledge of each other’s current position.

To elucidate the concept of the rendezvous problem we refer to the following two motivating examples. The first one is named the telephone problem and is formulated as follows: ``In each of two rooms, there are $n$ telephones randomly strewn about. They are connected in a pairwise fashion by $n$ wires. At discrete times $t = 0, 1, 2, \cdots$ players in each room pick up a phone and say ‘hello.’ They wish to minimize the time $t$ when they first pick up paired phones and can communicate.''~\cite{alpern2002rendezvous}. This problem is the spatial rendezvous on a complete graph, i.e. a graph where an edge exists between every pair of vertices.

The second motivating problem is the Mozart Café Rendezvous Problem: ``Two friends agree to meet for lunch at the Mozart Café in Vienna on the first of January, 2000. However, on arriving at Vienna airport, they are told there are three (or $n$) cafés with that name, no two close enough to visit on the same day. So each day each can go to one of them, hoping to find his friend there.''~\cite{alpern2010rendezvous}.

One variant of the rendezvous problem is the “symmetric rendezvous” where both players have equal capabilities and constraints or are following the same strategy~\cite{anderson1995rendezvous,yu1996agent,alpern1998symmetric,han2008improved}. In this case, the goal is to find a strategy that ensures both players meet at some vertex, regardless of their initial positions. Another variant is the “asymmetric rendezvous,” where the players have different capabilities or constraints and are following the same strategy~\cite{anderson1998asymmetric,alpern1999asymmetric,alpern2000asymmetric,alpern2000pure}. For example, one player may have limited mobility or restricted vision compared to the other player. In this case, the objective is to find a strategy that maximizes the probability of meeting at the rendezvous point while taking into account the differences in capabilities.

The importance of studying the rendezvous problem lies in its relevance to various real-world applications. For instance, in search and rescue missions, autonomous robots or drones may need to coordinate their movements to cover a large area efficiently and meet at a specific location to exchange information or resources. Similarly, in autonomous vehicle systems, vehicles may need to coordinate their routes and timing to avoid collisions and efficiently utilize shared resources like charging stations~\cite{dias2021swarm}.

By understanding and solving the rendezvous problem, researchers can develop efficient strategies for coordinating multiple agents without relying on direct communication. This can lead to improved efficiency, resource utilization, and safety in various domains. More generally, the problem is about coordinating action when the communication is prohibited or severely limited, see e.g.~\cite{pelc2012deterministic} for an overview.

For this work, we will use the following definition:

\begin{dfn}
	Given a graph, and a number $h\ge 1$, the \textbf{rendezvous task} is defined as the task in which $r\geq 2$ agents, placed uniformly randomly among the vertexes of the graph, move along edges of the graph $h$ times. They are successful if, after crossing the edges $h$ times, they are all positioned in the same vertex. They can  establish a strategy before being placed on the graph, but they are not allowed to communicate before each of the agents has completed their move.
 
    If $h = 1$ then the rendezvous task is called \textbf{single-step}.
\end{dfn}

Note that, for the single-step task, to allow strategies for which the agents can end their movement on the same vertex in which they started,  each of the vertexes of the graph must be connected with itself.

\subsection{Domination on Graphs}
\label{sec:dominationDef}

The domination number is a fundamental concept in graph theory~\cite{haynes2013fundamentals,haynes2023domination}. It quantifies the minimum number of vertices needed to control, or dominate, all other vertices in a graph. In essence, it identifies the smallest set of vertices where each vertex either belongs to the set or is adjacent to at least one member of the set. Formally, consider a graph $G = (N, E)$ with a vertex set $N$ and an edge set $E$. A dominating set $D$ is a subset of $N$ such that every vertex in $N$ is either part of $D$ or is adjacent to a vertex in $D$. The domination number, denoted as $\gamma(G)$, is defined as the minimum cardinality of any dominating set $D$.

Various variants of the domination number exist, each imposing specific conditions on dominating sets. Common variants include:
\begin{enumerate}
	\item Total Domination Number~\cite{cockayne1980total,henning2013total}: This requires that every vertex in the graph is adjacent to at least one vertex in the dominating set, ensuring direct control over each vertex.
	\item Independent Domination Number~\cite{Berge1962,ore1962theory,goddard2013independent}: In this variant, the dominating set must also be an independent set, meaning no two vertices in the set are adjacent. This variant seeks dominating sets without redundant control.
	\item Connected Domination Number~\cite{sampathkumar1979connected}: Here, the dominating set must induce a connected subgraph, ensuring an efficient path between any two vertices in the set for comprehensive control over the entire graph.
\end{enumerate}

The domination number and its variants have various applications and usefulness in different domains. For network design, the domination number can be used to optimize network design problems, such as determining the minimum number of sensors or routers required to monitor or control a network effectively. By finding an optimal dominating set, resources can be allocated efficiently while maintaining network connectivity. In the area of facility location, this quantity can help in determining the optimal location of facilities, such as hospitals, fire stations, or surveillance cameras, to ensure maximum coverage and control over a given area. By minimizing the domination number, the cost and resources required for facility placement can be reduced. The domination number can be also applied in social network analysis to identify influential individuals or groups. By finding dominating sets in social networks, it is possible to identify key players who have control or influence over a large portion of the network.

The domination problem til now has been investigated only from the static global point of view, where the entity was deciding which vertices can be used to dominate the graph. In our work, we propose a new distributed approach to the problem and treat it as a dynamic task of a group of mobile agents who start at not known in advance locations, and want to organize themselves without communication so that they dominate as large part of a graph as possible. To be more specific, in this work, we use the following definition:

\begin{dfn}
	Given a graph, and a number $h\ge 1$, the \textbf{domination task}  is defined as the task in which $r\geq 2$ agents, placed uniformly randomly among the vertexes of the graph, move along edges of the graph $h$ times. They try to dominate as many vertexes as they can, counting the dominated points only after all the agents completed their moves.
 
  A vertex is said to be dominated when an agent occupies it or it is connected with an edge to a vertex where there is an agent.
  
     The agents can  establish a strategy before being placed on the graph, but they are not allowed to communicate before each of the agents has completed their move.
 
    If $h = 1$ then the domination task is called \textbf{single-step}.
\end{dfn}

Note that, for the single-step task, to allow strategies for which the agents can end their movement on the same vertex in which they started, each vertex of the graph must be connected with itself.

\subsection{Considered Scenarios and Definitions}
\label{sec:graphs}

In Fig.~\ref{fig:graphs} we present some of the graph scenarios with which we dealt. With the name $n$-gon we refer to the graphs described by a polygon with $n$ vertexes, with the name $n$-line we refer to the graphs described by a line connecting $n$ vertexes, while with the name $n$-line curly we refer to a $n$-line graph where the extrema are connected with themselves.

At exception of the $n$-line curly graphs, for which we provided a specific definition, for all the other graphs we included the word "curly" in their name when each of their vertexes is connected with itself.

\begin{figure*}
	\centering
	\subfigure[Double triangle]{\label{fig:4vertexes_diamond_flat}\includegraphics[width=35mm]{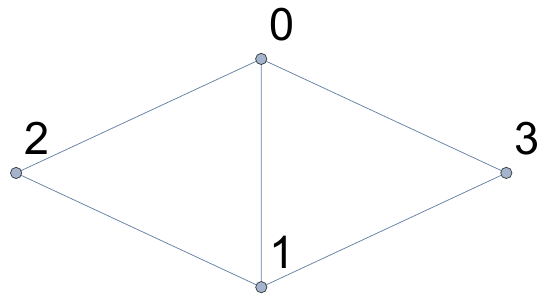}}
	\subfigure[Tetrahedron]{\label{fig:4vertexes_diamond}\includegraphics[width=35mm]{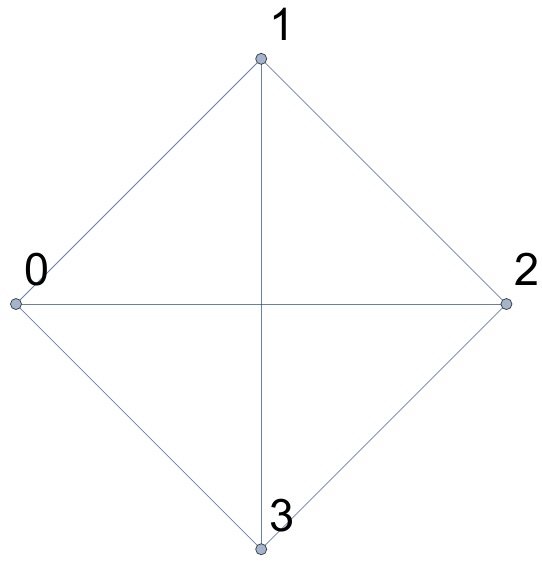}}
	\subfigure[Square curly]{\label{fig:4vertexes}\includegraphics[width=35mm]{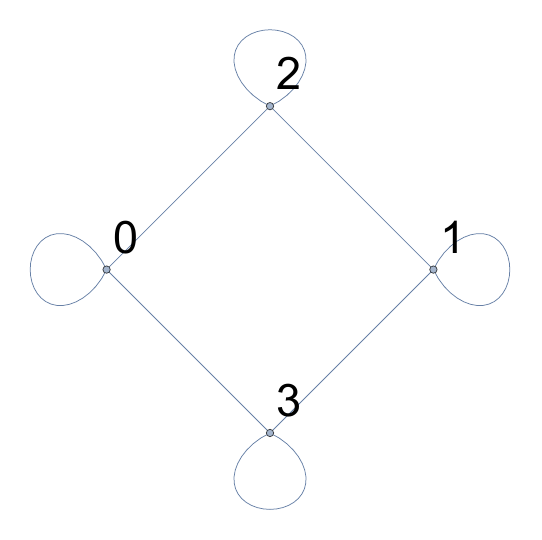}}
	\subfigure[Pentagon curly]{\label{fig:5vertexes}\includegraphics[width=35mm]{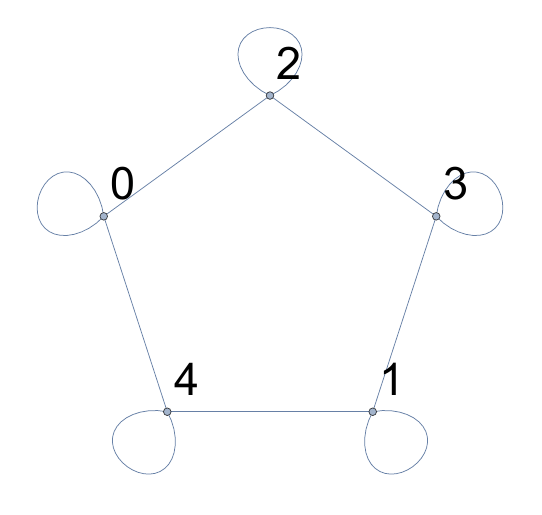}}
	\subfigure[Spike]{\label{fig:5vertexes_spike}\includegraphics[width=35mm]{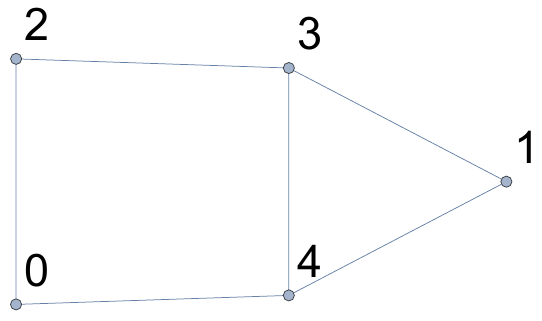}}
	
	\subfigure[Spike curly]{\label{fig:5vertexes1}\includegraphics[width=35mm]{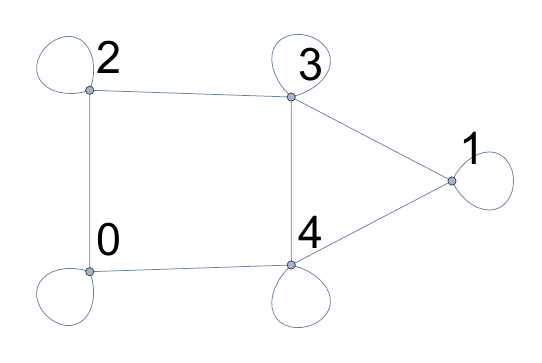}}
	\subfigure[Arrow]{\label{fig:5vertexes_arrow}\includegraphics[width=35mm]{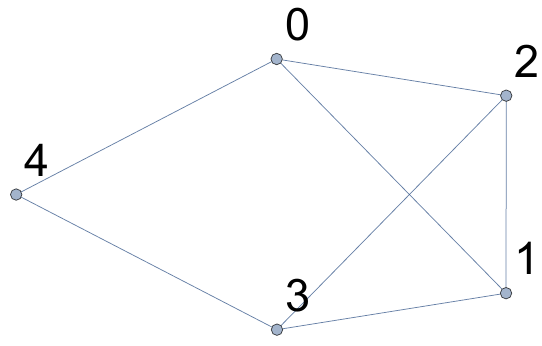}}
	\subfigure[Arrow curly]{\label{fig:5vertexes2}\includegraphics[width=35mm]{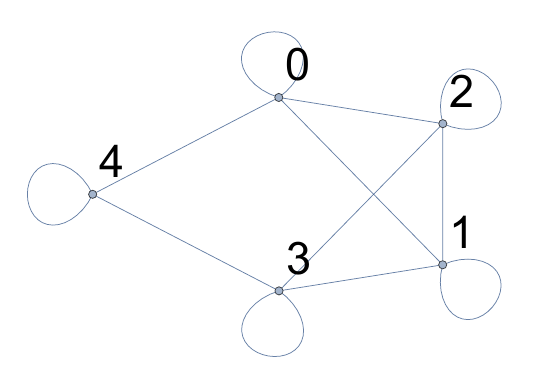}}
	\subfigure[Clamp]{\label{fig:6vertexes_clamp}\includegraphics[width=35mm]{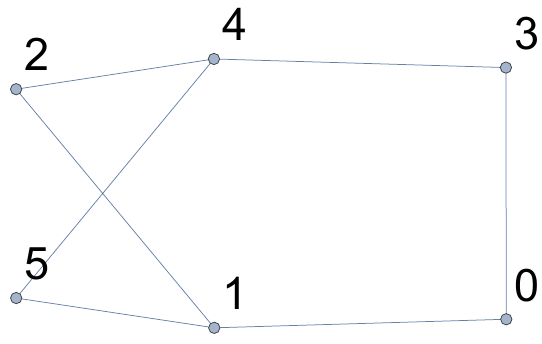}}
	\subfigure[Hat]{\label{fig:6vertexes_hat}\includegraphics[width=35mm]{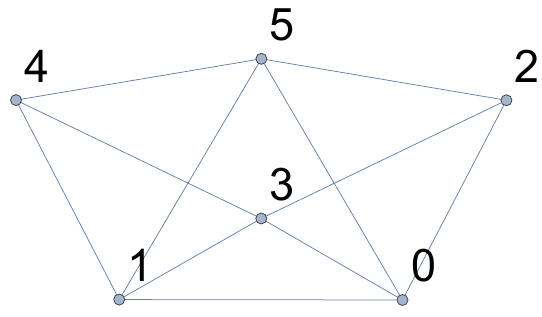}}
	
	\subfigure[House]{\label{fig:6vertexes_house}\includegraphics[width=35mm]{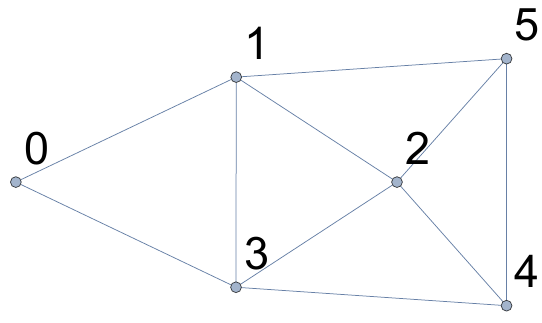}}
	\subfigure[Pyramid, double]{\label{fig:6vertexes_pyramid_double}\includegraphics[width=35mm]{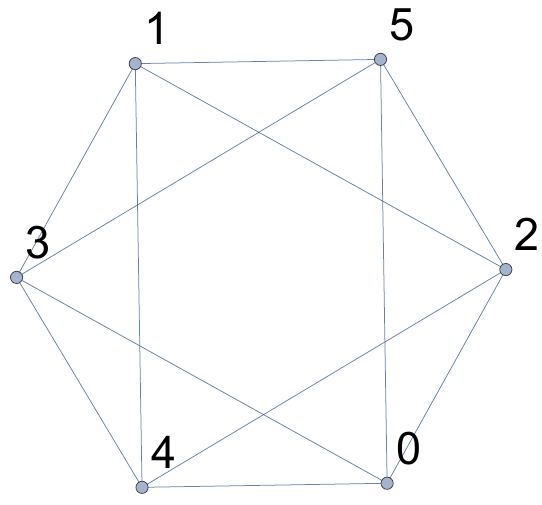}}
	\subfigure[Caltrop]{\label{fig:6vertexes_caltrop}\includegraphics[width=35mm]{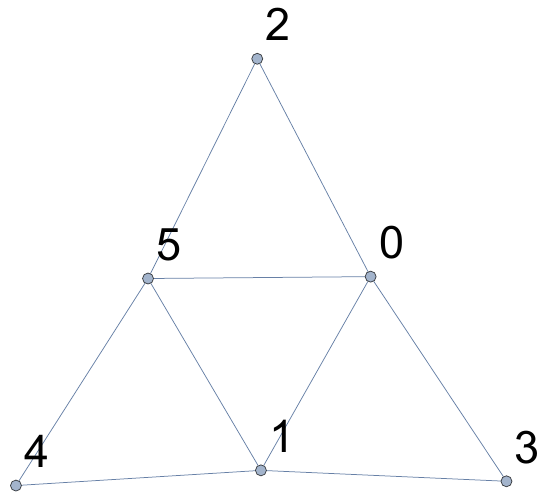}}
	\subfigure[Cube]{\label{fig:8vertexes_cube}\includegraphics[width=35mm]{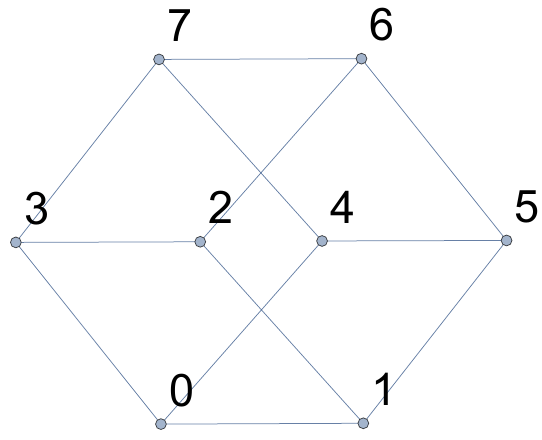}}
	
	\caption{\label{fig:graphs} Some of the graphs analysed in this work.}
\end{figure*}

Let $r$ denote the number of agents performing a given distributed task, let $N$ be the set of all the nodes of the considered graph and $n$ be the number of nodes, $n = \Ab{N}$. Let us introduce the following notation:
\begin{enumerate}
	\item $p_i(a|x)$ denotes the probability that the $i$-th agent will go to the node $a \in N$ when starting from the node $x\in N$;
	\item $S_i:=\{p_i(a|x)\}_{a,x \in N}$ is the classical strategy for the $i$-th agent;
	\item $p_i(a) := \frac{1}{n}\sum_{x\in N} p_i(a|x)$ is the probability to find the $i$-th agent in the node $a\in N$;
	\item $\tilde{S}_i:=\{p_i(a)\}_{a \in N}$.
\end{enumerate}
Note that $\tilde{S}_i$ is uniquely determined by the set $S_i$. A strategy is called deterministic (for all the agents) if and only if all the elements of  $S_i$, $ \forall i\in \{1,...,r\}$, are either zeros or ones. Given the set of strategies adopted by the agents $\{S_i\}_{i\in \{1,...,r\}}$, the strategies $S_i$'s are called \textit{symmetric} if they are all equal among each other.

\subsection{Classical and Quantum Probabilities and Bell Games}
\label{sec:probs}

The classical joint probability distribution $P(a,b|x,y)$ for obtaining outcomes $a$ and $b$ in response to inputs $x$ and $y$ is expressed as a sum over a hidden variable $\lambda$, viz.
\begin{equation}
	\label{eq:pabxy_lhv}
	P(a,b|x,y) = \sum_{\lambda} P(\lambda) \cdot P_{A|X,\Lambda}(a|x,\lambda) \cdot P_{B|Y,\Lambda}(b|y,\lambda).
\end{equation}
Each term in the sum corresponds to a specific hidden state $\lambda$ and is weighted by the probability $P(\lambda)$ of observing that hidden state. The conditional probabilities $P_{A|X,\Lambda}(a|x,\lambda)$ and $P_{B|Y,\Lambda}(b|y,\lambda)$ represent the likelihood of obtaining outcomes $a$ and $b$ for Alice and Bob, respectively, given their inputs $x$ and $y$ and the hidden state $\lambda$. The probability distribution $P_{\Lambda}(\lambda)$ characterizes the likelihood of different hidden internal states. The normalization condition $\sum_{\lambda \in \Lambda} P_{\Lambda}(\lambda) = 1$ ensures that the probabilities are properly scaled. This formulation captures the classical statistical description of the behavior of the devices and is often called a local hidden variables (LHV) model, which assumes that the measurement outcomes are determined by pre-existing properties of the system that are independent of the measurement settings

A bipartite quantum probability distribution of joint measurement results refers to the statistical distribution of outcomes obtained from a joint measurement performed on a bipartite quantum system. The bipartite system is described by a density operator $\rho$. The joint measurement performed by Alice and Bob is typically represented by a set of positive operator-valued measures (POVMs), denoted as ${M(a,x)}$ ($N(b,y)$), where $a$ ($b$) and $x$ ($y$) represent the measurement outcomes and settings of Alice (Bob). These POVMs satisfy the completeness relation $\sum_a M(a,x) = \openone$ ($\sum_b N(b,y) = \openone$), where $\openone$ is the identity operator~\cite{nielsen2010quantum}.

The bipartite quantum probability distribution $P(a,b|x,y)$ gives the probability of obtaining measurement outcomes $a$ and $b$ when Alice and Bob use measurement settings $x$ and $y$, respectively. It is calculated using the Born rule, which states that the probability is given by:
\begin{equation}
	\label{eq:pabxy_q}
	P(a,b|x,y) = \Tr [\rho M(a,x) \otimes N(b,y)].
\end{equation}
Here, $\Tr [\rho M(a,x) \otimes N(b,y)]$ represents the trace of the product of the density operator $\rho$ and the corresponding POVM elements $M(a,x)$ and $N(b,y)$. The bipartite quantum probability distribution captures the correlations between Alice’s and Bob’s measurement outcomes. These correlations can exhibit various phenomena such as entanglement, non-locality, or classical correlations depending on the quantum state $\rho$. We denote the set of all possible quantum bipartite probability distributions by $\mathcal{Q}$.

A bipartite game is a linear functional over probability distribution of a form
\begin{equation}
	\label{eq:bipartiteGame}
	\mathcal{B}[P(a,b|x,y)] = \sum_{a,b,x,y} \beta_{a,b,x,y} P(a,b|x,y)
\end{equation}
for some coefficients $\{\beta_{a,b,x,y}\}$. Eq.~\eqref{eq:bipartiteGame} can be easily generalized to more parties.

It is worth noting that in some cases, the bipartite quantum probability distribution may be incompatible with classical theories. This means that it cannot be explained solely by LHV models. If this is the case for a given game,                   then such a game is called a Bell game or non-local game~\cite{horodecki2009quantum,brunner2014bell}. Its maximal value over classical distributions  is called the classical bound and denoted $C$, and its maximal value over quantum distributions is called the Tsirelson bound and denoted by $Q$~\cite{cirel1980quantum}, $Q > C$.



In the case in which the agents participating in the domination or rendezvous tasks have access to quantum resources, they share a chosen quantum state before being placed on the graph. Once on the graph, they can measure their own state depending on the node they are at and choose which node to move to based on the outcome of their local measurement.

\subsection{Distributed Tasks as Bell Games}
\label{sec:tasks_as_games}

The success probability of a distributed task involving two agents, Alice and Bob, can be expressed as a linear functional of joint probabilities $P(a,b|x,y)$, where successes are assigned a coefficient of $1$ and failures are assigned a coefficient of $0$. This approach allows us to quantify the likelihood of success based on the joint probabilities of the agents’ actions and observations.

To understand this concept better, let’s break down the notation used. In this point of view the probabilities $P(a,b|x,y)$ as defined in sec.~\ref{sec:probs} and represent the joint probability distribution of Alice’s action $a$, Bob’s action $b$, Alice’s observation $x$, and Bob’s observation $y$. Thus, it captures the probabilistic relationship between the actions and observations of both agents. We note that the observation is performed depending on the agent's local situation, like information about its position in a graph, and the action depends on the measurement result.

Indeed, in a distributed task, Alice and Bob may need to coordinate their actions or make decisions based on their individual observations. The success of the task depends on how well they align their actions and observations to achieve the desired outcome. By considering the joint probabilities $P(a,b|x,y)$, we can analyze the likelihood of success in different scenarios. To express the success probability as a linear functional, we assign coefficients to each possible outcome (success or failure) based on its desirability. In this case, successes are assigned a coefficient of $1$, indicating their importance in determining overall success. Failures, on the other hand, are assigned a coefficient of $0$ since they do not contribute to the success of the task. Let $S(a,b,x,y)$ be a relevant scoring function of a distributed task. Using the coefficients for all possible outcomes and their corresponding joint probabilities, we obtain a linear combination that represents the score as:
\begin{equation}
	\label{eq:pabxy_game_task}
	\sum_{a,b,x,y} S(a,b,x,y) \times P(a,b|x,y),
\end{equation}
which is in the form of a bipartite game. This approach allows us to quantify the likelihood of success in distributed tasks by considering the joint probabilities and assigning appropriate coefficients to different outcomes. By manipulating these coefficients, we can prioritize certain outcomes or adjust the importance of different actions and observations in determining success or normalize events to interpret the score as a success probability of a task.



This approach allows us to analyze the likelihood of success based on the agents’ actions and observations and compare the capabilities possible in classical  and quantum  probability distributions. We say that we have a quantum advantage for a given task if the score is a Bell game with $Q > C$.




\section{Numerical Optimization Methods}
\label{sec:methods}

In this section, we describe the SDP method in sec.~\ref{sec:SDP}. This optimization technique is used to formulate two tools allowing for calculations of the success measure in distributed tasks, viz. NPA and see-saw. The former, discussed in sec.~\ref{sec:NPA}, provides upper, and the latter, discussed in sec.~\ref{sec:seesaw}, lower bounds on successes, respectively.

\subsection{Semi-definite Programming}
\label{sec:SDP}

SDP is a mathematical optimization technique that extends the concepts of linear programming to handle matrices and, in particular, positive semidefinite matrices. It addresses problems where the goal is to optimize a linear objective function subject to linear equality and semidefinite inequality constraints. SDP finds applications in diverse fields~\cite{vandenberghe1996semidefinite}, including quantum information~\cite{Skrzypczyk2023,mironowicz2023semi,tavakoli2023semidefinite}. Advancements in algorithmic design and improvements in computational power have significantly enhanced the efficiency of SDP solvers. Modern interior-point methods~\cite{potra2000interior} have proven to be effective in solving large-scale SDP, making them a versatile tool for addressing complex optimization problems~\cite{andersen2000mosek}.  For $m, n \in \mathbb{N}_{+}$ the primal optimization task of SDP is:
\begin{align}
	\begin{split}
		\text{minimize } &\null c^{T} \cdot x \\
		\text{subject to } &\null F(x) \succeq 0, \\
	\end{split}
\end{align}
where $c \in \mathbb{R}^m$, $F(x) := F_0 + \sum_{i = 1}^{m} x_i F_i$, $F_i \in \mathbb{R}^{n \times n}$, and $x \in \mathbb{R}^m$ is the variable. The so-called dual of this problem with a symmetric matrix variable $Z \in \mathbb{R}^{n \times n}$ is
\begin{align}
	\begin{split}
		\text{maximize } &\null - \Tr \left[ F_0 Z \right] \\
		\text{subject to } &\null \Tr \left[ F_i Z \right] = c_i, \text{ for } i = 1, \cdots, m,\\
		&\null Z \succeq 0.
	\end{split}
\end{align}

\subsection{The Navascu\'es-Pironio-Ac\'{\i}n Method}
\label{sec:NPA}

The NPA method, introduced by Navascués et al. in 2007~\cite{navascues2007bounding}, is a mathematical framework used to study quantum correlations and entanglement in quantum information theory. It provides a systematic approach to analyze the behavior of quantum systems and their correlations. Recall that in quantum mechanics, entanglement refers to the phenomenon where two or more particles become correlated in such a way that their states cannot be described independently. These correlations are stronger than any classical (or local) correlations and play a crucial role in various quantum information processing tasks.

This technique aims to quantify and characterize these quantum correlations by constructing a hierarchy of SDPs. The hierarchy is defined by introducing a sequence of operators that capture the correlations between systems under consideration and are associated with different levels of the hierarchy. Moving up the hierarchy, at each level, new sequences of operators are introduced that constrain or restrict additional correlations beyond those captured at lower levels. The NPA method provides a systematic way to compute these moment operators at each level and study the properties of quantum correlations rigorously. In this work, we employ the so-called almost quantum~\cite{Navascues2015a} level $1+ab$, and more precise, but also more computationally demanding, level $2$.

Let us now explain NPA in more detail. As we mentioned, quantum probability distributions $P(a,b|x,y)$, see Eq.~\eqref{eq:pabxy_q}, are challenging to characterize due to their inherent complexity. Recall that they belong to the set $\mathcal{Q}$ if certain conditions are met, namely, if there exists a Hilbert space $\mathcal{H}$, a state (vector) $\ket{\psi}$, and a set of operators (measurements) $\{E^a_x, E^b_y\}_{a,b,x,y}$ satisfying specific criteria. These criteria include Hermitian properties of operators, orthogonality of different measurement outcomes, normalization conditions, and commutativity of measurements between Alice and Bob. The probability distribution is then expressed as an inner product involving measurements of the quantum states. However, characterizing $\mathcal{Q}$ without quantum formalism poses a challenge. The NPA method provides a solution by defining a hierarchy $\{\mathcal{Q}_k\}_{k=1}^{\infty}$ of SDP problems. This hierarchy offers increasingly accurate approximations of $\mathcal{Q}$, with higher levels yielding more precise solutions, albeit at the expense of increased computational complexity. Ultimately, the hierarchy converges to the quantum set $\mathcal{Q}$~\cite{navascues2008convergent}, reflecting the effectiveness of the NPA method in characterizing quantum correlations.

A sequence of operators is formed by concatenating projective measurement operators. For instance, $E^1_2 E^3_2 F^2_1 E^1_1$ represents a sequence of four operators. Exploiting the commutativity of Alice's operators $E^a_x$ with Bob's operators $F^b_y$, we can rearrange the sequence as $E^1_2 E^3_2 E^1_1 F^2_1$. Given that $E^a_x E^{a^{\prime}}_x = 0$ and $F^b_y F^{b^{\prime}}_y = 0$ for $a \neq a^{\prime}$ and $b \neq b^{\prime}$, and leveraging the commutation property, we obtain, for instance, $E^2_1 F^3_3 E^1_1 = E^2_1 E^1_1 F^3_3 = 0$ since $E^2_1$ and $E^1_1$ are orthogonal. Moreover, as $E^a_x$ are projectors, $(E^a_x)^k = E^a_x$ for any $k \geq 1$, and similarly for $F^b_y$~\cite{mironowicz2018applications,mironowicz2023semi}.

The length of a sequence of operators $S$ denotes the minimal number of projectors required to express it. The null sequence corresponds to the identity operator, denoted by $\openone$, with its length defined as $0$. Consider an $n$-element set $\mathcal{S}$ of sequences of operators, such as $\mathcal{S}_{1+ab} = \left\{ \openone, E^a_x, F^b_y, E^a_x F^b_y \right\}_{\substack{a \in A, b \in B \\ x \in X, y \in Y}}$. By employing the NPA method, one constructs a hierarchy of relaxations using different choices of the set of sequences $\mathcal{S}$. Specifically, a set $\mathcal{S}_k$ comprises all sequences of operators $\{E^a_x, E^b_y\}$ with a length at most $k$. It follows that $\mathcal{S}_{1+ab} = \mathcal{S}_1 \cup \{ E^a_x E^b_y \}$, where $\mathcal{S}_{1+ab}$ is defined accordingly.

The NPA method aims to characterize a given bipartite probability distribution as quantum, implying the existence of a realization with a state $\ket{\psi}$ and projective measurements $\{ E^a_x, F^b_y \}$. This realization satisfies, for all settings $x \in X$ and $y \in Y$ and outcomes $a \in A$ and $b \in B$, the equation $P(a,b|x,y) = \bra{\psi} E^a_x F^b_y \ket{\psi}$. Considering operators $O_i, O_j$ in the set $\mathcal{S}$, let $\Gamma_{O_i, O_j}$ be defined as $\bra{\psi} O_i^{\dagger} O_j \ket{\psi}$. Consequently, $\Gamma_{E^a_x, E^b_y} = P(a,b|x,y)$ and $\Gamma_{\openone, \openone} = 1$. This equation defines an $n \times n$ matrix, where rows and columns are indexed by elements of $\mathcal{S}$, forming the so-called moment matrix $\Gamma$.

The elements of $\Gamma$ satisfy linear constraints: $\Gamma_{i,j} = \Gamma_{k,l}$ if and only if $O_i^{\dagger} O_j = O_k^{\dagger} O_l$, and $O_i^{\dagger} O_j = 0$ implies $\Gamma_{i,j} = 0$. For $v \in \mathbb{C}^n$ and $V = \sum_j v_j O_j$, it follows that $v^\dagger \Gamma v = \sum_{i,j} v_i^{*} \Gamma_{i,j} v_j = \sum_{i,j} v_i^{*} \bra{\psi} O_i^{\dagger} O_j \ket{\psi} v_j = \bra{\psi} V^{\dagger} V \ket{\psi} = \left| \langle V \ket{\psi} \right|^2 \geq 0$, thereby ensuring $\Gamma \succeq 0$. We obtain the relaxation by requiring the existence of such $\Gamma$ instead of the existence of states and operators realizing $P(a,b|x,y)$. This method is particularly useful for analyzing complex quantum systems and can provide valuable insights into the nature of quantum correlations and phenomena. The results obtained from NPA, as it constitutes a relaxation, provide an upper bound on the exact solution of the quantum optimization problem.

\subsection{The See-Saw Method}
\label{sec:seesaw}

An alternative approach to optimization over quantum distributions $\mathcal{Q}$, when we have certain restrictions on the dimensions of operators included in the setup, is the so-called see-saw~\cite{pal2010maximal} method. Again, we consider bipartite quantum probabilities given by expressions of the form of Eq.~\eqref{eq:pabxy_q}. It is easy to see that expressions that are linear combinations of probabilities given by formulas of the form of Eq.~\eqref{eq:bipartiteGame}, such as Bell-type operators, are linear in each term, \textit{viz.} the quantum state, Alice's measurement operators, and Bob's measurement operators. However, since the optimization must be done over all these groups of variables, the whole expression is non-linear.

The key idea of the see-saw method is the alternating use of a series of optimizations for which two of the expressions constituting Eq.~\eqref{eq:pabxy_q} are treated as constants, and the third of them is a variable subjected to optimization by SDP methods with the objective function being the considered Bell operator. For this purpose, in the first iteration, two of the terms take fixed, often randomly chosen, values.

To be more specific, see-saw allows to maximize a given Bell expression, which serves as a score function       in the considered bipartite games. This Bell expression, denoted as $\mathcal{B}[P(a,b|x,y)]$, quantifies the correlations between measurement outcomes $a$ and $b$ for settings $x$ and $y$. It is computed as a weighted sum of bipartite probabilities $P(a,b|x,y)$, as defined in Eq. ~\eqref{eq:pabxy_q}, where the coefficients $\beta_{a,b,x,y}$ weight the contributions of each outcome to the overall score.

At each iteration of the see-saw method, the quantum state and measurements of both Alice and Bob are optimized to maximize the Bell expression. This involves three main optimization steps: first, the state $\rho$ is optimized while holding the measurements of Alice and Bob constant; next, all Alice's measurements $M(a,x)$ are optimized while keeping the state and Bob's measurements fixed; finally, Bob's measurements $N(b,y)$ are optimized while the state and Alice's measurements are held constant.

During each iteration, the value of the Bell expression is computed using the updated quantum state and measurements. If the improvement in the Bell expression value falls below a predefined threshold or the maximum number of iterations is reached, the optimization process terminates. Through this iterative procedure, the see-saw method efficiently explores the space of possible quantum correlations, even though is not guaranteed to be converging towards the optimal solution that globally maximizes the Bell expression. Nonetheless, by providing explicit quantum state and measurements, it sheds light on the quantum nature of the underlying system. Indeed, note that while the NPA method provides an approximation of the quantum set of probabilities from its exterior, meaning a relaxation, the see-saw method finds direct representations of the quantum state and measurements implementing the found distribution. Thus, any solution obtained by the see-saw method is lower bound on the exact solution of the quantum optimization problem, complementing the results from NPA.

\section{Properties of Symmetric Deterministic Strategies for Distributed Tasks}
\label{sec:theorems}

In sec.~\ref{sec:symDetRendez} we prove lemmas showing that deterministic symmetric strategies are sufficient for consideration of LHV models for the rendezvous task, and then in sec.~\ref{sec:detDomination} we prove an analogous result for the domination task.

\subsection{Symmetric Deterministic Strategies for Rendezvous on Graphs}
\label{sec:symDetRendez}

Given any set of strategies $\{S_i\}_{i\in \{1,...,r\}}$ we have the success probability $W(\{S_i\}_{i\in \{1,...,r\}})$ is given by
\begin{equation}
	W(\{S_i\}_{i\in \{1,...,r\}}) = \sum_{a\in N} \prod_{i=1} ^r  p_i (a)=\tilde{W}(\{\tilde{S}_i\}_{i\in \{1,...,r\}})
\end{equation}
so the success probability is completely determined by the set $\{\tilde{S}_i\}_{i\in \{1,...,r\}}$.\\

\begin{lemma}
	\label{lem:symNonsym}
	Symmetric strategies for the rendezvous task are at least as good as non-symmetric strategies.
\end{lemma}
\begin{proof}
	We have that
	\begin{equation}
		\begin{aligned}
			\tilde{W}(\{\tilde{S}_i\}_{i\in \{1,...,r\}}) &= \sum_{a\in N} \prod_{i=1} ^r  p_i (a) \leq \\
            & \prod_{i=1} ^r \sqrt[r]{\sum_{a\in N} p_i(a)^{r} }    \leq\\
			& \max_{i \in {1,...,r}} \sum_{a\in N} p_i(a)^{r} = \tilde{W}(\{\tilde{S}_{i^*},...,\tilde{S}_{i^*}\}),
		\end{aligned}
	\end{equation}
	where $i^*$ is the value of $i$ for which the maximum is reached. Here for the first estimation we used the Hölder inequality for sums~\cite{holder1889uber}.
	
	So, given any set of strategies adopted by the $r$ agents, the winning probability is upper bounded by the success probability of the symmetric strategy which employs a suitable strategy chosen among the set of the given strategies.
\end{proof}

\begin{lemma}
	\label{lem:symDeterministic}
	Symmetric deterministic strategies for the rendezvous task are at least as good as symmetric non-deterministic strategies.
\end{lemma}
\begin{proof}
	Given any non-deterministic strategy $S =\{p (a|x)\}_{a,x\in N}$, repeated for each of the agents, and the associated set $\tilde{S}=\{p(a)\}_{a\in N}$, we have that it is always possible to build a deterministic strategy, repeated for each of the agents, which has a greater success probability. In fact, if $S_1$ is not deterministic, there is at least one input node $x_0$ for which there are at least two nodes $a_1$ and $a_2$ for which $p(a_1|x_0), p(a_2|x_0)> 0$.
	
	Without loss of generality, we consider the case in which $p(a_1)\geq p(a_2)$. A better strategy $S' = \{p' (a|x)\}_{a,x\in N} $, repeated for each of the agents, is built by assigning to the node $x_0$ a probability $p'(a_1|x_0) = p(a_1|x_0)+p(a_2|x_0)$ and $p'(a_2|x_0) = 0$, and the same probabilities in all the other cases. This is true because, for any $b,c \in N$, $r \geq 2$ with $p(b)\geq p(c)$ and $0 <\epsilon \leq p(c)$, we have that
	\begin{equation}
		p(b)^r + p(c)^r < (p(b)+\epsilon)^r + (p(c)-\epsilon)^r.
	\end{equation}
	Then, it is possible to repeat this procedure till you obtain a deterministic strategy.
	
	So any non-deterministic symmetric strategy has a lower success probability of at least one symmetric deterministic strategy.
\end{proof}

Putting together the previous two results from Lemmas~\ref{lem:symNonsym} and~\ref{lem:symDeterministic} we get the Corrolary~\ref{cor:detSymSufficient}.

\begin{cor}
	\label{cor:detSymSufficient}
	It is possible to find the optimal strategy for the rendezvous problem by investigating only deterministic symmetric strategies. 
\end{cor}

Note that these results may not be applicable in other variants of the rendezvous task, such as when agents cannot be placed in the same initial nodes, in which case it is still possible to find the optimal solution by investigating only deterministic strategies, because the success probability remains a linear function of the probabilities, which are subject to linear constraints.

\subsection{Deterministic strategies for Domination on Graphs}
\label{sec:detDomination}

Let us now consider the graph domination task.

\begin{lemma}
    \label{lem:Domination}
	Deterministic strategies for the domination task are at least as good as non-deterministic strategies.
\end{lemma}
\begin{proof}
	Let's consider $K(\{S_j\}_{j \in \{1,...,r\}})$ the score, i.e. the average number of dominated points, of a given strategy $\{S_j\}_{j \in \{1,...,r\}}$. A set of strategies that provides score at least as high as $K(\{S_j\}_{j \in \{1,...,r\}})$ can be obtained by replacing the strategy of the first agent with the following deterministic strategy.
	
	Knowing the strategies $\{S_j\}_{j \in \{2,...,r\}})$, for any fixed input $x\in N$ for the $1$-th agent the best move would be to go to the accessible node which would provide the highest increment in the score, knowing what is the probability that any node is going to be dominated by any of the other agents. Assigning this deterministic strategy to the first agent the score is greater or equal to the case in which they were using any other strategy.
	
	Then, repeating the procedure to each of the agents and updating the strategies employed by them, we obtain the thesis.
\end{proof}

\section{Results}
\label{sec:Results}

The content of the section includes the results of rendezvous with quantum entanglement in sec.~\ref{sec:rendezvousResults}, domination with quantum entanglement in sec.~\ref{sec:dominationResults}, and a detailed analysis of selected cases in sec.~\ref{sec:detailed}.

In the following will be presented the scenarios where we applied this technique and the tables that contain the results obtained.

In tables the columns "Classical" denote the maximum values which can be achieved by classical strategies. The columns "Random" denote the value obtained when the parties choose in a random uniform way which node to reach from the starting position. The values associated with the column "NPA" are calculated for the level $1+ab$ when they appear without $ \cdot ^*$, otherwise, they are calculated for level $2$. The column "Adv." contains the gain obtained by using quantum resources compared to the optimal classical strategy. The quantity is expressed in percents, and is calculated according to the following expression:
\begin{equation}
\label{adv}
	\frac{Q-R}{C-R} - 1 = \frac{Q-C}{C-R},
\end{equation}
where $C$ denotes the classical value, i.e. the highest value that can be achieved when the parties can use only classical strategies. $R$ denotes the average success probability for the rendezvous task and the average number of dominated nodes for the domination task, when the agents choose in a random uniform way which node to reach from the starting position, selecting the node to reach among the nodes that can be reached while respecting the rules of the given task. While $Q$ denotes the value achieved when the agents employ the strategy found with the see-saw technique.

When, for a given scenario, within at least $100$ runs of see-saw algorithm we didn't find an average success probability greater than the one obtainable employing classical strategies and which doesn't match the upper bound found with the level $2$ of the NPA hierarchy, we say that the results are inconclusive. 

Inconclusive results can happen because the level of the NPA hierarchy investigated, or the number of runs of see-saw algorithm, or the dimension of the investigated quantum states or the rank of the projectors generated with see-saw were too low.

Applying see-saw algorithm we investigated only projectors of rank $1$ for all the cases of $3$ agents and up to rank $2$ for the cases involving $2$ agents. For the dimension of the generated quantum  states, we have always chosen the lowest one determined by the rank of the generated projectors.

Note that in the case of the rendezvous task, due to lemma ~\ref{lem:symDeterministic}, to find the best classical value it is sufficient to investigate only all the symmetric  deterministic strategies.  In the case of the domination task, due to lemma ~\ref{lem:Domination},  to find the best classical value it is sufficient to investigate only all the deterministic strategies. 

Both for the domination task and rendezvous task, with the additional constraint that the agents can't start from the same positions, to find the classical value it is sufficient to study only deterministic strategies. This is true because the average success probability for the rendezvous task and the average number of dominated points for the domination task are still linear function of probabilities which are subject to linear constraints.

When dealing with quantum strategies constrained to be symmetric, we considered only strategies for which all the agents can perform only the same set of measurements on the state and that to them is associated the same same marginal probability distribution, adjusting accordingly the constraints when performing the see-saw technique.

\subsection{Rendezvous with quantum entanglement}
\label{sec:rendezvousResults}

We will now present the results obtained when the agents are dealing with the single-step rendezvous task.

Tab.~\ref{Rend_2_any} illustrates the scenarios where a quantum advantage is observed for two agents undertaking the rendezvous task, with the freedom to commence from any location, including the same initial position for both parties.

\begin{table}[htb!]
	\centering
	\caption{Results associated with the single-step rendezvous task, two agents case, when the agents can start from any position.}
	\label{Rend_2_any}
	\begin{tabular}{|c|c|c|c|c|}
		\hline
		Name & Random & Classical & NPA & Adv. [\%]\\
		\hline
		\begin{tabular}[c]{@{}c@{}}  tetrahedron \\ Fig.~\ref{fig:4vertexes_diamond} \end{tabular} &    $\frac{1}{4}$ & $\frac{5}{8}$ &  $0.64506^*$ & 5 \\ 
		\hline
		\begin{tabular}[c]{@{}c@{}} square \\curly, Fig.~\ref{fig:4vertexes}  \end{tabular} &    $\frac{1}{4}$  &  $\frac{5}{8}$ &  $0.64506^*$ & 5 \\ 
		\hline
		\begin{tabular}[c]{@{}c@{}} pentagon \\curly, Fig.~\ref{fig:5vertexes} \end{tabular}  &   $\frac{1}{5}$ & $\frac{13}{25}$  &  $0.53009^*$  & 3 \\ 
		\hline
		\begin{tabular}[c]{@{}c@{}} arrow \\Fig.~\ref{fig:5vertexes_arrow}  \end{tabular} &   0.20667 & $\frac{13}{25}$ &   $0.52051^*$ & 0.1 \\ 
		\hline
		\begin{tabular}[c]{@{}c@{}} clamp \\ Fig.~\ref{fig:6vertexes_clamp} \end{tabular}  &   0.18827 & $\frac{7}{18}$ &  $0.40063^*$ & 6 \\ 
		\hline
		\begin{tabular}[c]{@{}c@{}}hat \\ Fig.~\ref{fig:6vertexes_hat} \end{tabular} &    0.17207 & $\frac{5}{9}$ &  $0.58333 \approx\frac{7}{12}$ & 7 \\ 
		\hline
		\begin{tabular}[c]{@{}c@{}} house \\ Fig.~\ref{fig:6vertexes_house}  \end{tabular} &   0.18210 & $\frac{5}{9}$ &  $0.58333 \approx \frac{7}{12}$ & 7 \\ 
		\hline
		\begin{tabular}[c]{@{}c@{}} caltrop \\ Fig.~\ref{fig:6vertexes_caltrop} \end{tabular} &   0.20833 & $\frac{5}{9}$ &  $0.58333 \approx\frac{7}{12}$ & 8 \\ 
		\hline
		\begin{tabular}[c]{@{}c@{}} cube \\ Fig.~\ref{fig:8vertexes_cube}  \end{tabular} &   $\frac{1}{8}$ & $\frac{5}{16}$ &   $0.32253^*$ &  5 \\ 
		\hline
		\begin{tabular}[c]{@{}c@{}} triangle \\ 3-gon \end{tabular}  & $\frac{1}{3}$ & $\frac{5}{9}$  &  $0.58333 \approx\frac{7}{12}$ &   13 \\ 
		\hline
		\begin{tabular}[c]{@{}c@{}}pentagon \\ 5-gon \end{tabular}  & $\frac{1}{5}$ & $\frac{9}{25}$ &  0.38090 &  13 \\ 
		\hline
		\begin{tabular}[c]{@{}c@{}} hexagon \\ 6-gon   \end{tabular}  & $\frac{1}{6}$  &$\frac{5}{18}$  & 0.29167 &  13 \\ 
		\hline
		\begin{tabular}[c]{@{}c@{}}  heptagon \\ 7-gon \end{tabular}  & $\frac{1}{7}$ & $\frac{13}{49}$ &  0.27864 & 11 \\ 
		\hline
		\begin{tabular}[c]{@{}c@{}} ennagon \\ 9-gon \end{tabular}  &  $\frac{1}{9}$ & $\frac{17}{81}$ & 0.21887  &  9 \\ 
		\hline
		\begin{tabular}[c]{@{}c@{}}decagon \\ 10-gon \end{tabular} &   $\frac{1}{10}$ &  $\frac{9}{50}$  &   0.19045 & 13 \\ 
		\hline
		\begin{tabular}[c]{@{}c@{}} 11-gon   \end{tabular}  &   $\frac{1}{11}$  &  $\frac{21}{121}$  &  0.17998 & 8 \\ 
		\hline
		\begin{tabular}[c]{@{}c@{}} 13-gon \end{tabular} &   $\frac{1}{13}$ &  $\frac{25}{169}$ &  0.15273 & 7 \\ 
		\hline
		\begin{tabular}[c]{@{}c@{}} 3-line curly \end{tabular}&   $\frac{1}{3}$ &  $\frac{5}{9}$ &  $0.58333 \approx\frac{7}{12}$ & 13 \\ 
		\hline
		5-line curly &   $\frac{1}{5}$ & $\frac{9}{25}$ &   0.38090 & 13 \\ 
		\hline
		7-line curly &   $\frac{1}{7}$ & $\frac{13}{49}$ &   0.27864 & 11 \\ 
		\hline
	\end{tabular}
\end{table}

In contrast, Tab.~\ref{Rend_2_Nany_all} delineates situations where this advantage occurs when agents are refrained from initiating from identical positions. This dichotomy underscores the robustness of quantum strategies across diverse starting conditions, offering a notable advantage over classical counterparts.

\begin{table}[htb!]
	\centering
	\caption{Results associated with the single-step rendezvous task, two agents case, when the agents can't start in the same positions.}
	\label{Rend_2_Nany_all}
	\begin{tabular}{|c|c|c|c|c|}
		\hline
		Name & Random & Classical & NPA & Adv. [\%]\\
		\hline
		\begin{tabular}[c]{@{}c@{}}  tetrahedron \\ Fig.~\ref{fig:4vertexes_diamond} \end{tabular} & $\frac{2}{9}$ & $\frac{1}{2}$ & $0.53333^* \approx  \frac{8}{15}$ & 12 \\ 
		\hline
		\begin{tabular}[c]{@{}c@{}} square \\curly, Fig.~\ref{fig:4vertexes}  \end{tabular}   &  $\frac{2}{9}$ &  $\frac{1}{2}$ & $0.53333^* \approx \frac{8}{15}$ & 12 \\ 
		\hline
		\begin{tabular}[c]{@{}c@{}} pentagon \\curly, Fig.~\ref{fig:5vertexes} \end{tabular} & $\frac{1}{6}$  & $\frac{2}{5}$ & $0.41316^*$ & 6 \\ 
		\hline
		\begin{tabular}[c]{@{}c@{}} arrow \\Fig.~\ref{fig:5vertexes_arrow}  \end{tabular}  & $0.16667$  &  $\frac{2}{5}$  & $0.40490^*$ & 2 \\ 
		\hline
		\begin{tabular}[c]{@{}c@{}} clamp \\ Fig.~\ref{fig:6vertexes_clamp} \end{tabular}  & 0.13704   & $\frac{4}{15}$ & $0.28229^*$ & 12 \\ 
		\hline
		\begin{tabular}[c]{@{}c@{}}hat \\ Fig.~\ref{fig:6vertexes_hat} \end{tabular}  & 0.15093  & $\frac{7}{15}$ & $0.50000 \approx \frac{1}{2}$ & 11 \\ 
		\hline
		\begin{tabular}[c]{@{}c@{}} house \\ Fig.~\ref{fig:6vertexes_house}  \end{tabular}  & 0.15463  & $\frac{7}{15}$ & $0.50000 \approx \frac{1}{2}$ & 11 \\ 
		\hline
		\begin{tabular}[c]{@{}c@{}} caltrop \\ Fig.~\ref{fig:6vertexes_caltrop} \end{tabular}  & $\frac{7}{40}$  & $\frac{7}{15}$  & $0.50000 \approx \frac{1}{2}$ & 11 \\ 
		\hline
		\begin{tabular}[c]{@{}c@{}} cube \\ Fig.~\ref{fig:8vertexes_cube}  \end{tabular}  & $\frac{2}{21}$  & $\frac{3}{14}$ & $0.22857^*$ & 12 \\ 
		\hline
	\end{tabular}
\end{table}

Expanding on this comparison, Tab.~\ref{Rend_2_Nany_sym} delves into the case where agents can't start in the same positions, and they are allowed to use only symmetric strategies. Here, the introduction of symmetric strategies sheds light on the intricate interplay between quantum and classical approaches.

\begin{table}[htb!]
	\centering
	\caption{Results associated with the single-step rendezvous task, two agents case, when the agents can't start in the same positions and they are constrained to adopt only symmetric strategies. Note that the first nine cases has identical values as those in Tab.~\ref{Rend_2_Nany_all}.} 
	\label{Rend_2_Nany_sym}
	\begin{tabular}{|c|c|c|c|c|}
		\hline
		Name & Random & Classical & NPA & Adv. [\%]\\
		\hline
		\begin{tabular}[c]{@{}c@{}}  tetrahedron \\ Fig.~\ref{fig:4vertexes_diamond} \end{tabular}  & $\frac{2}{9}$ & $\frac{1}{2}$ & $0.53333^* \approx\frac{8}{15}$ &  12 \\ 
		\hline
		\begin{tabular}[c]{@{}c@{}} square \\curly, Fig.~\ref{fig:4vertexes}  \end{tabular} & $\frac{2}{9}$ & $\frac{1}{2}$ &  $0.53333^* \approx\frac{8}{15}$ &  12 \\ 
		\hline
		\begin{tabular}[c]{@{}c@{}} pentagon \\curly, Fig.~\ref{fig:5vertexes} \end{tabular} & $\frac{1}{6}$ & $\frac{2}{5}$  &  $0.41316^*$ & 6 \\ 
		\hline
		\begin{tabular}[c]{@{}c@{}} arrow \\Fig.~\ref{fig:5vertexes_arrow}  \end{tabular} & $\frac{1}{6}$ & $\frac{2}{5}$  & $0.40490^*$ & 2 \\ 
		\hline
		\begin{tabular}[c]{@{}c@{}} clamp \\ Fig.~\ref{fig:6vertexes_clamp} \end{tabular} & 0.13704 & $\frac{4}{15}$ & $0.28229^*$ & 12 \\ 
		\hline
		\begin{tabular}[c]{@{}c@{}}hat \\ Fig.~\ref{fig:6vertexes_hat} \end{tabular} & 0.15093 & $\frac{7}{15}$ & $0.50000 \approx \frac{1}{2}$ & 11 \\ 
		\hline
		\begin{tabular}[c]{@{}c@{}} house \\ Fig.~\ref{fig:6vertexes_house}  \end{tabular} & 0.15463 & $\frac{7}{15}$ & $0.50000 \approx \frac{1}{2}$ & 11 \\ 
		\hline
		\begin{tabular}[c]{@{}c@{}} caltrop \\ Fig.~\ref{fig:6vertexes_caltrop} \end{tabular} & $\frac{7}{40}$ & $\frac{7}{15}$ & $0.50000 \approx \frac{1}{2}$ & 11 \\ 
		\hline
		\begin{tabular}[c]{@{}c@{}} cube \\ Fig.~\ref{fig:8vertexes_cube}  \end{tabular} &  $\frac{2}{21}$  & $\frac{3}{14}$ &  $0.22857^*$ & 12 \\ 
		\hline
		\begin{tabular}[c]{@{}c@{}} triangle \\ 3-gon \end{tabular} & $\frac{1}{4}$ & $\frac{1}{3}$ & $0.50000 \approx \frac{1}{2}$ & 200 \\ 
		\hline
		\begin{tabular}[c]{@{}c@{}}pentagon \\ 5-gon \end{tabular} & $\frac{1}{8}$ &  $\frac{1}{5}$ & $0.25000 \approx \frac{1}{4}$ & 67 \\ 
		\hline
		\begin{tabular}[c]{@{}c@{}} hexagon \\ 6-gon   \end{tabular} & $\frac{1}{10}$ &  $\frac{2}{15}$  & $0.20000 \approx \frac{1}{5}$ & 200 \\ 
		\hline
		\begin{tabular}[c]{@{}c@{}} heptagon \\ 7-gon \end{tabular} & $\frac{1}{12}$ & $\frac{1}{7}$ & $0.16667 \approx \frac{1}{6}$ & 40 \\ 
		\hline
		\begin{tabular}[c]{@{}c@{}} ennagon \\ 9-gon \end{tabular} & $\frac{1}{16}$ &  $\frac{1}{9}$ &  $0.12500 \approx \frac{1}{8}$ & 29 \\ 
		\hline
		\begin{tabular}[c]{@{}c@{}}decagon \\ 10-gon \end{tabular} & $\frac{1}{18}$ &  $\frac{4}{45}$ & $0.11111 \approx \frac{1}{9}$ & 67 \\ 
		\hline
		\begin{tabular}[c]{@{}c@{}} 11-gon   \end{tabular} & $\frac{1}{20}$ & $\frac{1}{11}$ & $0.10000 \approx\frac{1}{10}$ & 22 \\ 
		\hline
		\begin{tabular}[c]{@{}c@{}} 13-gon \end{tabular} & $\frac{1}{24}$  & $\frac{1}{13}$ & $ 0.08333 \approx \frac{1}{12}$ & 18 \\ 
		\hline
		\begin{tabular}[c]{@{}c@{}} 3-line  curly \end{tabular} &  $\frac{1}{4}$ & $\frac{1}{3}$ &  $0.50000 \approx \frac{1}{2}$ &  200 \\ 
		\hline
		5-line curly&  $\frac{1}{8}$  & $\frac{1}{5}$ &  $0.25000 \approx \frac{1}{4}$ & 67 \\ 
		\hline
		7-line curly&  $\frac{1}{12}$  & $\frac{1}{7}$ &  $0.16667 \approx \frac{1}{6}$ & 40 \\ 
		\hline
	\end{tabular}
\end{table}

Remarkably, classical strategies exhibit limitations in maximizing success probabilities under symmetric conditions unless, as shown in sec.~\ref{sec:symDetRendez}, any starting position is allowed. On the other hand, quantum strategies continue to demonstrate a clear advantage in both cases, as we elucidate further in sec.~\ref{ssec:triangle}.

The juxtaposition of these findings underscores the nuanced dynamics at play in the rendezvous task, emphasizing the pivotal role of starting conditions and strategy symmetry. While classical strategies attain their peak when any starting position is permissible, quantum strategies transcend these limitations, showcasing their adaptability and efficacy even under symmetric constraints. This nuanced understanding is pivotal for harnessing the full potential of quantum strategies in real-world applications of multi-agent coordination.

When dealing with the rendezvous task with three agents we have found no advantage when they can exploit quantum resources, both when they can start in any position, both when they can't start in the same position, both with and without the restriction of using symmetric strategies for the following graphs: square curly (Fig.\ref{fig:4vertexes}), double triangle (Fig.\ref{fig:4vertexes_diamond_flat}), cube (Fig.\ref{fig:8vertexes_cube}), tetraedron (Fig.\ref{fig:4vertexes_diamond}), arrow curly (Fig.\ref{fig:5vertexes2}), clamp (Fig. \ref{fig:6vertexes_clamp}), hat (Fig.\ref{fig:6vertexes_hat}), house (Fig.\ref{fig:6vertexes_house}). While it is not conclusive for pentagon curly (Fig.\ref{fig:5vertexes}) and arrow (Fig.\ref{fig:5vertexes_arrow}) when the agents can't start in the same positions. We found no advantage for any case for the graphs $n$-lines and $n$-lines curly for $4 \leq n \leq 8$. 

We have no reason to exclude that, in the three agent case, an advantage could be found by using more computational resources, when dealing with the inconclusive cases, or increasing the number of nodes of the graphs studied.

Let us denote the considered cases given in Tabs~\ref{Rend_2_any}, \ref{Rend_2_Nany_all}, and~\ref{Rend_2_Nany_sym} by $\mathcal{A}$, $\mathcal{B}$, and $\mathcal{C}$, respectively. Recall that they refer to cases when any starting position is allowed ($\mathcal{A}$), when the agents can’t start in the same positions, also with non-symmetric strategies ($\mathcal{B}$), and when the agents can’t start in the same positions and they use symmetric strategies ($\mathcal{C}$).

We observe that if for a given graph $1+ab$ is enough to determine its exact Tsirelson bound for one of the cases ($\mathcal{A}$, $\mathcal{B}$, or $\mathcal{C}$) then it is enough for all cases ($\mathcal{A}$, $\mathcal{B}$, and $\mathcal{C}$). What is more, for cycles in all cases the NPA level $1+ab$ is enough.

All the cases in $\mathcal{C}$ that do not appear in $\mathcal{B}$ are cases for which we found that there is no advantage when the agents share a quantum state, so they are cases for which the quantum correlations are exploited to break the symmetry between the agents, as will be described more in detail in sec.~\ref{ssec:triangle}. Additionally, all the cases in $\mathcal{B}$ present the same average success probabilities both for the classical and both for the quantum case as the ones that they present in $\mathcal{C}$, and so they are cases for which the classical value can be reached by symmetric strategies.

The graphs triangle, 3-line curly, hat, house, and caltrop have always the same value of classical and quantum for a given case $\mathcal{A}$ or $\mathcal{C}$. Except for triangle and 3-line curly, this holds also for case $\mathcal{B}$.

It is worth noting that, at least for the cases analyzed, if for a cycle the number of vertices is divisible by $4$ then there is no advantage in any case $\mathcal{A}$, $\mathcal{B}$, or $\mathcal{C}$.

\subsection{Domination with quantum entanglement}
\label{sec:dominationResults}

In Tabs~\ref{Domination_2_any} and~\ref{Domination_2_Nany} we present the results obtained when the agents deal with the graph domination task when they can exploit quantum resources.

\begin{table}[htb!]
	\centering
	\caption{Results associated with the single-step domination task, two agents case, when the agents can start from any position.}
	\label{Domination_2_any}
	\begin{tabular}{|c|c|c|c|c|}
		\hline
		Name & Random & Classical & NPA & Adv. [\%] \\
		\hline
		\begin{tabular}[c]{@{}c@{}} pentagon\\ curly, Fig.~\ref{fig:5vertexes} \end{tabular} & 4.2 & 4.64  & $4.67361^*$ & 8 \\ 
		\hline
		\begin{tabular}[c]{@{}c@{}} caltrop\\ Fig.~\ref{fig:6vertexes_caltrop} \end{tabular} &  5.458333 & 5.88889  & 5.916667 & 6 \\ 
		\hline
		\begin{tabular}[c]{@{}c@{}} spike \\ Fig.~\ref{fig:5vertexes_spike}
		\end{tabular} & 4.51333 & 4.92  & 4.93 & 2 \\ 
		\hline
		\begin{tabular}[c]{@{}c@{}}clamp\\  Fig.~\ref{fig:6vertexes_clamp}
		\end{tabular} & 4.94907 & 5.44444  & 5.45453 & 2 \\ 
		\hline
		pentagon & 4.2 & 4.6  & 4.67361 & 18 \\ 
		\hline
		hexagon & 4.50000 & 4.94445  &  5.0000 & 13 \\ 
		\hline
		heptagon & 4.71428  & 5.08163  & 5.15517 & 20 \\ 
		\hline
		octagon & 4.875 & 5.1875  &  5.23928 & 17 \\ 
		\hline
		\begin{tabular}[c]{@{}c@{}}  9-gon \end{tabular} & 5 &  5.24691  &  5.29434 & 19 \\ 
		\hline
		\begin{tabular}[c]{@{}c@{}}   10-gon
		\end{tabular} & 5.1 & 5.3  & 5.33680  & 18 \\ 
		\hline
		\begin{tabular}[c]{@{}c@{}}   11-gon
		\end{tabular} & 5.18182 & 5.39669  &  5.43395 & 17 \\ 
		\hline
		\begin{tabular}[c]{@{}c@{}}   12-gon
		\end{tabular} & 5.25 & 5.47222  &  5.5 & 13 \\ 
		\hline
		\begin{tabular}[c]{@{}c@{}}   13-gon
		\end{tabular} & 5.30769 & 5.50888  &  5.54543 & 18 \\ 
		\hline
		6-line curly & 4.11111 & 4.44445  &  4.44895 & 1 \\ 
		\hline
	\end{tabular}
\end{table}

\begin{table}[htb!]
	\centering
	\caption{Results associated with the single-step domination task, two agents case, when the agents can’t start in the same positions.}
	\label{Domination_2_Nany}
	\begin{tabular}{|c|c|c|c|c|}
		\hline
		Name & Random & Classical & NPA & Adv. [\%]\\
		\hline
		\begin{tabular}[c]{@{}c@{}} pentagon\\ curly, Fig.~\ref{fig:5vertexes} \end{tabular} & 4,27778 &  4.7  & 4.73987 & 9 \\ 
		\hline
		\begin{tabular}[c]{@{}c@{}}clamp\\  Fig.~\ref{fig:6vertexes_clamp}
		\end{tabular} & 5.01482 & 5.4   & $5.41210^*$ & 3 \\ 
		\hline
		\begin{tabular}[c]{@{}c@{}} caltrop\\ Fig.~\ref{fig:6vertexes_caltrop} \end{tabular} & 5.48750 & 5.86667  & 5.9 & 9 \\ 
		\hline
		\begin{tabular}[c]{@{}c@{}} spike\\ Fig.~\ref{fig:5vertexes_spike}
		\end{tabular} & 4.56944 & 4.9  & 4.9125 & 4 \\ 
		\hline
		\begin{tabular}[c]{@{}c@{}} pentagon \end{tabular}  & 4.25 & 4.5  & 4.59201 & 37 \\ 
		\hline
		\begin{tabular}[c]{@{}c@{}}hexagon\end{tabular}  & 4.60000 &  4.93333 & 5.00000 & 20 \\ 
		\hline
		\begin{tabular}[c]{@{}c@{}}heptagon \end{tabular} & 4.83333 &    5.09524 & 5.18103 &  33 \\ 
		\hline
		\begin{tabular}[c]{@{}c@{}}octagon
		\end{tabular} & 5 & 5.21429  & 5.27346 &  28 \\ 
		\hline
		\begin{tabular}[c]{@{}c@{}}  9-gon \end{tabular}  & 5.125 & 5.27778  & 5.33113 & 35 \\ 
		\hline
		\begin{tabular}[c]{@{}c@{}}   10-gon
		\end{tabular} & 5.22222 & 5.34444  &  5.37423 & 24  \\
		\hline
		\begin{tabular}[c]{@{}c@{}}   11-gon
		\end{tabular} & 5.3 & 5.43636  &  5.47735 & 30 \\ 
		\hline
		\begin{tabular}[c]{@{}c@{}}   12-gon
		\end{tabular} & 5.36364 & 5.51515  &  5.54545 & 20 \\ 
		\hline
		\begin{tabular}[c]{@{}c@{}}   13-gon
		\end{tabular} & 5.41667 & 5.51515  &  5.59088 & 29 \\ 
		\hline
	\end{tabular}
\end{table}

For the two agents case we found that there is not an advantage for any case associated with the following graphs: double triangle (Fig.~\ref{fig:4vertexes_diamond_flat}), tetraedron (Fig.~\ref{fig:4vertexes_diamond}), square curly (Fig.~\ref{fig:4vertexes}), arrow (Fig.~\ref{fig:5vertexes_arrow}), arrow curly (Fig.~\ref{fig:5vertexes2}), pyramid double (Fig.~\ref{fig:6vertexes_pyramid_double}), hat (Fig. \ref{fig:6vertexes_hat}), cube (Fig.~\ref{fig:8vertexes_cube}), $3$-line curly and $4$-line curly.

In the case of spike curly (Fig.~\ref{fig:5vertexes1}), we found that there is no advantage in using quantum resources when the agents can start from any position, while we found an advantage of $+0.224\%$ when they can't start in the same position, although with see-saw we didn't obtain the value found with the level $2$ of the NPA hierarchy.

When the agents can't start in the same positions we didn't find any advantage for all the $n$-line curly with $4\leq n \leq 8$.

When dealing with the three agents cases we found no advantage when the agents can't start in the same positions for all the $n$-line and $n$-line curly with $4\leq n \leq 8$ and, when they can start in any position not conclusive for $n$-line and $n$-line curly with $4\leq n \leq 7$ and for $n=8$ we found no advantage. In the case of $n$-gon, when the agents can start in the same positions, we found no advantage for $5\leq n \leq 8$, while, for the same graphs, when the agents can start in any positions, the results are inconclusive.

We have no reason to exclude that, in the three agent case, an advantage could be found by using more computational resources, when dealing with the inconclusive cases, or increasing the number of nodes of the graphs to study.

\subsection{Detailed analysis of selected cases}
\label{sec:detailed}

We will now consider several examples of strategies for the rendezvous and domination tasks to illustrate their working. In sec.~\ref{ssec:triangle} we discuss one of the cases from Tab.~\ref{Rend_2_Nany_sym} with particularly high advantage and discuss the role of quantum resources in this particular case. Next, in sec.~\ref{ssec:hat}, we discuss the quantum nature of the advantage for the rendezvous task with agents that are not symmetric. Finally, in sec.~\ref{ssec:pentagon} we explicitly illustrate a particular case of quantum advantage for the domination task.


\subsubsection{The 3-line curly graph with symmetric strategies for rendezvous}
\label{ssec:triangle}

Considering the case of the single-step rendezvous task with two agents, when the agents can’t start in the same positions and they are constrained to adopt only symmetric strategies, we now focus on the particular case of the graph \textit{3-line curly}, as reported in Tab.~\ref{Rend_2_Nany_sym}. We have that in the case, the best classical strategy achieves the average success probability $C=\frac{1}{3}$, while, when the players choose randomly the node to reach, they have an average success probability $R=\frac{1}{4}$. The quantum success probability is $Q=\frac{1}{2}$.


\begin{figure}
    \centering
    \includegraphics[width=35mm]{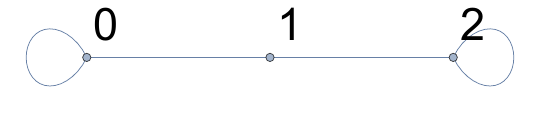}
    \caption{ Graph with a 3-line curly structure.}
    \label{fig:3_line_curly}
\end{figure}

For the 3-line curly graph, shown in Fig.~\ref{fig:3_line_curly}, if one of the parties starts in node $0$ and the other in node $1$, they win only if the former decides to stay in node $0$, and the other moves to node $0$, and this happens with probability $\frac{1}{4} = \frac{1}{2} \times \frac{1}{2}$. If one of the parties starts in node $0$ and the other in node $2$, they win only if both decide to move to node $1$, with probability $\frac{1}{4}$. Finally, if one of the parties starts in node $1$ and the other in node $2$ they win only if the former moves to node $2$, and the latter decides to stay in node $2$, which happens, for the random strategy, again with probability $\frac{1}{4}$.
Thus the average winning probability with random strategy is $\frac{1}{4}$.

The optimal symmetric deterministic strategy is the following: If the party is in node $0$ or $1$, then it should move to the possible node with the smaller label, i.e. $0$; if the party is in node $2$ it should move to the possible node with the largest label, i.e. $2$.


Stating this more explicitely, for the 3-line curly graph, when using the optimal deterministic strategy, if one of the parties starts in node $0$ and the other in node $1$, the former stays in node $0$, and the latter moves to node $0$, and they win. If one of the parties starts in node $0$ and the other in node $2$, then the former stays in node $0$, and the other stays in node $2$, and they fail. If one of the parties starts in node $1$ and the other in node $2$, the former moves to node $0$, and the latter stays in node $2$, and they fail. Thus, the winning probability is $\frac{1}{3}$.

In the quantum case, the strategy is found applying see-saw, finding the quantum state $\rho$ shared by the agents and the measurements performed by them, viz. $\{M(a,x)\}_{x \in \{0,1,2\}, a\in \{0,1\}}$ and $\{N(b,y)\}_{y \in \{0,1,2\}, b \in \{0,1\}}$, with $x$ and $y$ denoting the input node and $a$ and $b$ denoting the output of the measurement. It is sufficient to consider only measurements to which there are associated $2$ outcomes because in this graph each vertex is attached at most $2$ edges. In this example the joint probability distribution of the measurement results $\{P(a,b|x,y)\}_{a,b \in \{0,1\}, x,y \in \{0,1,2\}}$ can be transformed to a strategy $\{\tilde P(\tilde a,\tilde b|x,y)\}_{\tilde a,\tilde b \in \{0,1\}, x,y\in \{0,1,2\}}$ where $\tilde a$ and $\tilde b$ denote the output nodes. For example, if the agent is starting in the vertex $x\in \{0,1,2\}$, associating to the output $a=0$ the vertex labeled by the smallest number among the ones which can be reached starting from the input node $x$, and to the output $a=1$ the vertex labeled by the largest number among the ones which can be reached starting from the input node $x$.
We have the following quantum strategy $\{P(a,b|x,y)\}_{a,b \in \{0,1\},x,y\in \{0,1,2\}}$:
\begin{align}
    & P(1,0|0,0)= P(0,1|0,0) = P(0,0|1,0) = P(1,1|1,0) \notag\\&= P(1,0|2,0) = P(0,1|2,0)  = P(0,0|0,1) = P(1,1|0,1)\notag\\& = P(1,0|1,1) = P(0,1|1,1) = P(0,0|2,1) = P(1,1|2,1) \notag\\& = P(1,0|0,2) = P(0,1|0,2) = P(0,0|1,2) = P(1,1|1,2) \notag\\&= P(1,0|2,2) = P(0,1|2,2) = 0.5,
\end{align}
and $0$ otherwise.

So, considering only the cases for which $x\neq y $, we have that the average success probability provided by this strategy is 
\begin{align}
	Q &=\frac{ P(0,0|1,0) +P (0, 1|2, 0)+P (0, 0|0, 1)}{6}\notag\\&+\frac{P (1, 1|2, 1)+P (1, 0|0, 2)+P (1, 1|1, 2) }{6}=0.5.
\end{align}
where $6$ is the number of possible inputs for the agents for which they are not starting in the same node. So, applying Eq.~\eqref{adv}, we obtain that the advantage is $ 200\%$.


\subsubsection{The \textit{hat} graph with non-symmetric strategies for rendezvous}
\label{ssec:hat}

\begin{figure}
    \centering
    \includegraphics[width=50mm]{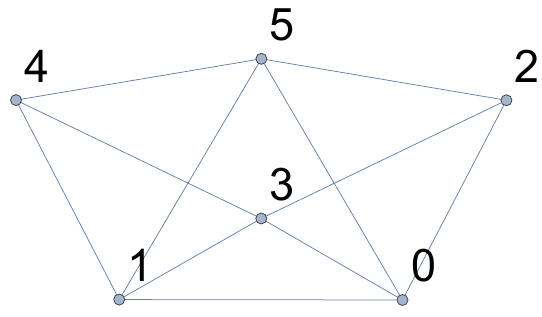}
    \caption{Graph with the "hat" structure.}
    \label{fig:6vertexes_hat_repeated}
\end{figure}

We will now discuss the hat graph shown in Fig.~\ref{fig:6vertexes_hat}, and repeated here in Fig.~\ref{fig:6vertexes_hat_repeated} for convenience.


In the case in which Alice and Bob can't start in the same position and without the symmetric strategy constraint, the optimal deterministic strategy, leading to the result stated in Tab.~\ref{Rend_2_Nany_all}, is the following: Alice starting from nodes $0$ or $3$ or $4$ or $5$ goes to node $1$; from node $1$ and $2$ goes to node $3$. The same strategy for Bob. This covers $12$ possible initial positions. Since we consider the case in which the agents start at different positions, we have $30$ possible initial locations. The parties win when Alice starts from nodes $\{0,3,4,5\}$ and Bob from nodes $\{0,3,4,5\}$, as then they meet at node $1$. This covers $12$ possible initial positions. The other winning possibilities are when they meet at node $3$, which happens in $2$ cases when Alice starts from node $1$ and Bob from node $2$ and vice versa. The average success probability in this case is $\frac{14}{30} = \frac{7}{15}$.

One of the quantum strategies giving the optimal value $0.5$ uses the entangled state $\frac{1}{\sqrt{2}} (\ket{00} + \ket{33})$ on Hilbert space $\mathbb{C}^4 \otimes \mathbb{C}^4$. For $x = 0$ the measurements of Alice are the following:
\begin{subequations}
	\begin{equation}
		M(0,0) = \proj{0} + 0.25 \proj{1} + 0.25 \proj{2},
	\end{equation}
	\begin{equation}
		M(1,0) = M(3,0) = 0.25 (\proj{1} + \proj{2}),
	\end{equation}
	\begin{equation}
		M(2,0) = 0.25 (\proj{1} + \proj{2}) + \proj{3}.
	\end{equation}
\end{subequations}
For $x = 1$ the measurements are:
\begin{subequations}
	\begin{equation}
		\begin{aligned}
			M(0,1) =& 0.75 \proj{0} + 0.25 (\proj{1} + \proj{2} + \proj{3}) \\ &+ \alpha (\kb{0}{3} + \kb{3}{0}),
		\end{aligned}
	\end{equation}
	\begin{equation}
		\begin{aligned}
			M(1,1) =& 0.25 (\proj{0} +\proj{1} + \proj{2}) + 0.75 \proj{3} \\ &- \alpha (\kb{0}{3} + \kb{3}{0}),
		\end{aligned}
	\end{equation}
 and
	\begin{equation}
		M(2,1) = M(3,1) = M(1,0).
	\end{equation}
\end{subequations}
Then, for $x = 2$ Alice uses the measurements:
\begin{subequations}
	\begin{equation}
		\begin{aligned}
			M(0,2) =& 0.75 \proj{0} + \frac{1}{3} (\proj{1} + \proj{2}) + 0.25 \proj{3} \\ &+ \alpha (\kb{0}{3} + \kb{3}{0}),
		\end{aligned}
	\end{equation}
	\begin{equation}
		\begin{aligned}
			M(1,2) =& 0.25 \proj{0} + \frac{1}{3} (\proj{1} + \proj{2}) + 0.75 \proj{3} \\ &- \alpha (\kb{0}{3} + \kb{3}{0}),
		\end{aligned}
	\end{equation}
	\begin{equation}
		M(2,2) = \frac{1}{3} (\proj{1} + \proj{2}),
	\end{equation}
\end{subequations}
and $M(3,2) = 0$. The measurements for $x = 3$ and for $x = 5$ are
\begin{subequations}
	\begin{equation}
		\begin{aligned}
			M(0,3) =& M(0,5) = 0.25 (\proj{0} + \proj{1} + \proj{2}) \\ & + 0.75 \proj{3}+ \alpha (\kb{0}{3} + \kb{3}{0}),
		\end{aligned}
	\end{equation}
	\begin{equation}
		\begin{aligned}
			M(1,3) =& M(1,5) = 0.75 \proj{0} + \\ &0.25 (\proj{1} + \proj{2} + \proj{3}) \\ &- \alpha (\kb{0}{3} + \kb{3}{0}),
		\end{aligned}
	\end{equation}
 and
	\begin{equation}
		M(2,3) = M(3,3) = M(2,5) = M(3,5) = M(1,0).
	\end{equation}
\end{subequations}
Finally, for $x = 4$ the measurement is given by
\begin{subequations}
	\begin{equation}
		M(0,4) = \proj{0} + \frac{1}{3} (\proj{1} + \proj{2}),
	\end{equation}
	\begin{equation}
		M(1,4) = M(2,4) = \frac{1}{3} (\proj{1} + \proj{2}) + \proj{3},
	\end{equation}
\end{subequations}
and $M(3,4) = 0$. 

The measurements of Bob are the same, and thus the quantum strategy is symmetric.

Let us now discuss some of the probabilities which occur in this case. For instance, let Alice start at node $0$ and Bob at node $1$. With probability $0.5$ Alice goes to node $1$ or $3$, and Bob with probability $0.5$ goes to node $0$ or $3$, but due to quantum steering~\cite{schrodinger1935discussion,wiseman2007steering,ramanathan2018steering,uola2020quantum} which is an effect of the entanglement, the probability of both parties arriving at node $3$  is $\frac{3}{8} > \frac{1}{2} \times \frac{1}{2}$. A similar situation happens for all other $23$ pairs of possible of settings, excluding the $6$ discussed below.

A different type of movement occurs in the following $6$ cases. The first two cases are when one of the parties starts at node $0$ and the other at node $4$. Then with probability $0.5$, they both move to node $1$, and with probability $0.5$ they both move to node $3$, so they win with probability $1$. Analogous action is performed when parties are at nodes $1$ and $2$ and the parties meet either at node $0$ or at node $3$. Similarly for starting pair of nodes $3$ and $5$ the parties meet either at $0$ or at $1$ with certainty.

The total winning probability is thus $\frac{3}{8} \times \frac{24}{30} + 1 \times \frac{6}{30} = 0.5$.

\subsubsection{The domination task on pentagon}
\label{ssec:pentagon}

Now, we consider the other of cooperation tasks, \textit{viz.} the domination.

In the case in which the agents can start from any position, the optimal classical strategy is shown in Fig.~\ref{fig:pentagon}. It can be noticed, that the strategy can be summarized as trying to concentrate one of the parties, as Alice, near nodes $1$ and $2$, and the other party near nodes $3$ and $4$. It can be calculated that the parties on average will dominate $4.6$ nodes.

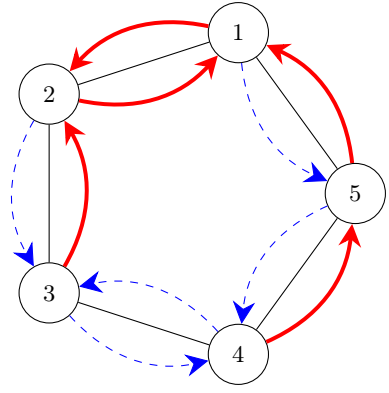
\begin{figure}
    \centering
    \begin{tikzpicture}[scale=1.5]
        \foreach \x/\y in {1/72,2/144,3/216,4/288,5/0}
            \node[circle, draw, minimum size=8mm] (n\x) at (\y:1.5cm) {\x};
            
        \foreach \x/\y in {1/2,2/3,3/4,4/5,5/1}
            \draw (n\x) -- (n\y);
            
        \foreach \x/\y in {1/2,2/1,3/2,4/5,5/1}
            \draw[red, -{Stealth[width=3mm,length=3mm]}, line width=1.5pt] (n\x) to [bend right] (n\y);
            
        \foreach \x/\y in {1/5,2/3,3/4,4/3,5/4}
            \draw[blue, -{Stealth[width=3mm,length=3mm]}, dashed] (n\x) to [bend right] (n\y);
    \end{tikzpicture}
    \caption{(color online) Pentagon with optimal classical strategies of Alice and Bob for domination. Red arrows denote moves of one of the parties, and blue arrows moves of the other party.}
    \label{fig:pentagon}
\end{figure}

One of the quantum strategies achieving the optimal value has the following properties:

First, consider the case in which Alice and Bob start in the same position, which happens in $5$ out of $25$ cases. In this case, we have that half of the times Alice moves in the so-called "clock-wise" direction~\cite{alpern1995rendezvous,flocchini1998sense,alpern2002rendezvous,pelc2012deterministic}, while Bob moves "counter-clock-wise", and in the other half of the times, Alice moves "counter-clock-wise" while Bob moves "clock-wise", dominating in both cases the $5$ nodes.

Now let us consider the case when the starting nodes of Alice and Bob are neighboring. This happens in $10$ out of $25$ possible initial position pairs.
In this case, with a probability about $p_1 \approx 0.3273$, they are moving in opposite directions, so that after their movement they will be separated by one node. Thus they will dominate all $5$ nodes; in Fig.~\ref{fig:pentagon_quantum} it is represented with green ultra-thin arrows.

With the same probability $p_1$, they will move towards each other, effectively exchanging their positions. They will dominate $4$ nodes; in Fig.~\ref{fig:pentagon_quantum} this move is represented with blue ultra-thick arrows.
With smaller probability $0.5 - p_1 \approx 0.1727$ both will move either "clock-wise" or "counter-clock-wise", and dominate $4$ nodes.

The remaining possible case is when the starting nodes of Alice and Bob are not neighboring. This happens in $10$ out of $25$ possible initial positions.
The quantum strategy implies that with probability $p_2 \approx 0.45224$, both parties move "clock-wise", and with the same probability they both move "counter-clock-wise", and thus after the movement they are still not neighboring. In that case, they dominate all $5$ nodes. In Fig.~\ref{fig:pentagon_quantum} this move is represented with yellow dashed ultra-thick arrows.

With smaller probability $0.5 - p_2 \approx 0.04773$, the first of the parties is moving "clock-wise", and the second "counter-clock-wise"; with the same probability the first is moving "counter-clock-wise" and the second "clock-wise". Depending on the exact starting node locations, the former situation will lead either to the moving of both parties to the same node,  dominating $3$ nodes, or moving to two neighboring nodes, dominating $4$ nodes; both situations happen with equal probability. In Fig.~\ref{fig:pentagon_quantum} this move is represented with red dashed and cyan dashed thick arrows.

Summarizing, the average number of dominated nodes is:
\begin{equation}
	\begin{aligned}
		&\frac{5}{25} \times 5+\frac{10}{25} \times \left[ p_1 \times 5 + (1-p_1) \times 4 \right] + \\ &\frac{10}{25} \times \left[ 2 p_2 \times 5 + (0.5 - p_2) \times 4 + (0.5 - p_2) \times 3 \right]  \\& \approx 4.6736.
	\end{aligned}
\end{equation}

\begin{figure}
    \centering
    \begin{tikzpicture}[->, >=Stealth, auto, node distance=3cm, main node/.style={circle, draw, minimum size=1.5em}, scale=1.5]
      \foreach \x/\y in {1/72,2/144,3/216,4/288,5/0}
        \node[main node] (\x) at (\y:1.5cm) {\x};
    
      \draw[green, bend left, ultra thin] (1) to node [font=\tiny, pos=0.5, above] {$p_1$} (2);
      \draw[green, bend right, ultra thin] (5) to node [font=\tiny, pos=0.5, above] {$p_1$} (4);
      \draw[blue, bend left, ultra thick] (5) to node [font=\tiny, pos=0.5, above] {$p_1$} (4);
      \draw[blue, bend left, ultra thick] (1) to node [font=\tiny, pos=0.5, above] {$p_1$} (5);
      \draw[yellow, bend left, dashed, ultra thick] (2) to node [font=\tiny, pos=0.5, above] {$p_2$} (1);
      \draw[yellow, bend right, dashed, ultra thick] (4) to node [font=\tiny, pos=0.5, above] {$p_2$} (3);
      \draw[red, bend left, dashed] (2) to node [font=\tiny, pos=0.5, right] {$0.5 - p_2$} (3);
      \draw[red, bend left, dashed] (4) to node [font=\tiny, pos=0.5, above] {$0.5 - p_2$} (3);
      \draw[cyan, bend left, dashed, thick] (3) to node [font=\tiny, pos=0.5, left] {$0.5 - p_2$} (2);
      \draw[cyan, bend left, dashed, thick] (5) to node [font=\tiny, pos=0.5, left] {$0.5 - p_2$} (1);
    \end{tikzpicture}
    \caption{(color online) Pentagon with optimal quantum strategies of Alice and Bob for domination. Arrows denote the moves of the parties in different cases, and the labels refer to probabilities of the movement; see the text in sec.~\ref{ssec:pentagon} for the description.}
    \label{fig:pentagon_quantum}
\end{figure}
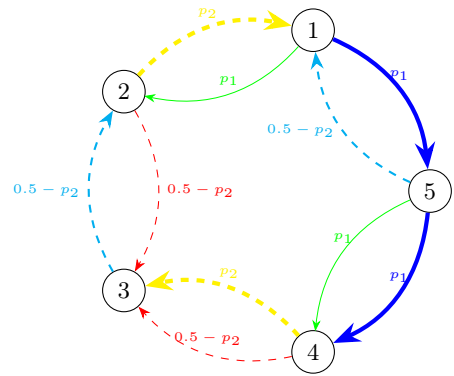

\section{Discussion}
\label{sec:discuss}

The two problems under scrutiny in this study, namely rendezvous on graphs and the introduced domination on graphs task for mobile agents, represent specific instances of a broader concept known as pattern formation~\cite{sugihara1996distributed,suzuki1999distributed,aggarwal1999hard,prencipe2000achievable}. Among the array of challenges in pattern formation research, the most akin task is dispersion on graphs~\cite{kshemkalyani2019efficient}, wherein the agents' objective is to efficiently spread throughout a defined territory, or graph exploration~\cite{dereniowski2015fast}. These problems are particularly significant as they find applications across a range of fields, encompassing robotics, networking, and distributed systems, and they grapple with fundamental issues pertaining to effective coordination and information exchange among mobile agents operating within a graph-based framework.

In the seminal paper preceding our current research~\cite{mironowicz2023entangled}, the focus was primarily on cubic graphs and cycles. Our present study extends this investigation to explore the broader applicability of the conclusions drawn from the advantage quantum entanglement offers in the context of rendezvous tasks for mobile agents. This work aims to establish the extent to which these findings can be generalized and the ubiquity of situations where quantum entanglement provides a strategic edge.

On the other hand, it is both surprising and somewhat disappointing that, on graphs where a quantum advantage was observed for two parties, we couldn't identify any gain for three agents. Consequently, investigating whether such an advantage is feasible, at least in certain scenarios, for more than two agents becomes a crucial avenue for further exploration. Uncovering such examples appears to be a non-trivial task. If it turns out to be unattainable, it would be valuable to discern and understand the underlying reasons for this limitation.

\section{Conclusions}
\label{sec:conclusions}

In our study, we introduced a novel distributed task for mobile agents, specifically the domination task, drawing inspiration from a classical concept in graph theory. Our exploration encompassed various scenarios across different types of graphs for both the domination task and another crucial multi-agent task, namely rendezvous. After obtaining some results about classical strategies, we examined cases involving two and three agents. Intriguingly, we consistently observed a quantum advantage for scenarios involving two parties, while no such advantage was found for three parties. This highlights the potential efficacy of the proposed method for tasks involving two agents and prompts further investigation into why it proves challenging to identify examples for three agents. This, along with the quest for general analytical formulae for various families of graphs, remains an open question for future explorations.

It is important to note that this approach offers the possibility to search for Bell inequalities associated with tasks with an arbitrary number of inputs and number of outputs associated with each input by building one of the suitable graphs and, after selecting a communication task, looking for gaps between the maximum success probabilities associated with classical and quantum strategies.

\section*{Acknowledgements}

We acknowledge partial financial support by the Foundation for Polish Science (IRAP project, ICTQT, contract No. 2018/MAB/5,
co-financed by EU within Smart Growth Operational Programme), from the Knut and Alice Wallenberg Foundation through the Wallenberg Centre for Quantum Technology (WACQT), and from NCBiR QUANTERA/2/2020 (www.quantera.eu) an ERA-Net cofund in Quantum Technologies under the project eDICT.
The numerical calculation we conducted using NCPOL2SDPA~\cite{wittek2015algorithm,Ncpol2sdpaGitHub}, and MOSEK Solver~\cite{andersen2000mosek}. P.M. acknowledge discussions with Jorge Quintanilla Tizon, Paul Strange, and Joshua Tucker from University of Kent.

\bibliographystyle{ieeetr}
\bibliography{rendezvous_domination_refs}

%
%
%

\end{document}